\documentclass[10pt,conference]{IEEEtran}
\usepackage[utf8]{inputenc}
\usepackage[T1]{fontenc}
\usepackage{microtype}%if unwanted, comment out or use option "draft"
\usepackage{multicol}
\usepackage{hyperref}
\hypersetup{
  colorlinks = true,
  linkcolor = blue,
  citecolor = red
}
\usepackage{xstring}
\usepackage{amsmath}
\usepackage{amsthm}
\usepackage{amssymb}
\usepackage{graphicx}
\usepackage{enumitem}
\usepackage{algorithm}
\usepackage{algpseudocode}
\usepackage{todonotes}
\usepackage{calc}

\theoremstyle{plain} \newtheorem{thm}{Theorem}
\theoremstyle{plain} \newtheorem{lem}{Lemma}
\theoremstyle{plain} 
\theoremstyle{plain} 

\newcommand{\sref}[1]{\StrCut{#1}{:}\pre\post%
                        \IfStrEqCase{\pre}{%
                           {alg}{\hyperref[#1]{Algorithm~\ref*{#1}}}%
                           {lem}{\hyperref[#1]{Lemma~\ref*{#1}}}%
                           {thm}{\hyperref[#1]{Theorem~\ref*{#1}}}%
                           {ln}{\hyperref[#1]{Line~\ref*{#1}}}%
                           {cor}{\hyperref[#1]{Corollary~\ref*{#1}}}%
                           {fig}{\hyperref[#1]{Figure~\ref*{#1}}}%
                           {sec}{\hyperref[#1]{Section~\ref*{#1}}}%
                           {app}{\hyperref[#1]{Appendix~\ref*{#1}}}%
                         }[ss]%
                       }
\newcommand{\mt}{\mathit}
\setlist[enumerate]{itemsep=1pt,parsep=0pt,before={\parskip=0pt}}

\begin{document}%\layout

\title{Byzantine Agreement with Unknown Participants and Failures}
%\date{}
% \author{\IEEEauthorblockN{Pankaj Khanchandani}
% \IEEEauthorblockA{ETH Zurich\\
% Zurich, Switzerland\\
% kpankaj@ethz.ch}
% \and
% \IEEEauthorblockN{Roger Wattenhofer}
% \IEEEauthorblockA{ETH Zurich\\
% Zurich, Switzerland\\
% wattenhofer@ethz.ch}
% }

%\thispagestyle{empty}
\author{\IEEEauthorblockN{Pankaj Khanchandani}
\IEEEauthorblockA{\textit{ETH Zurich} \\
Zurich, Switzerland \\
kpankaj@ethz.ch}
\and
\IEEEauthorblockN{Roger Wattenhofer}
\IEEEauthorblockA{\textit{ETH Zurich} \\
Zurich, Switzerland \\
wattenhofer@ethz.ch}
}
\maketitle

\begin{abstract}
    A set of mutually distrusting participants that want to agree on a common opinion must solve an instance of a Byzantine agreement problem. These problems have been extensively studied in the literature. However, most of the existing solutions assume that the participants are aware of $n$ --- the total number of participants in the system --- and $f$ --- an upper bound on the number of Byzantine participants. In this paper, we show that most of the fundamental agreement  problems can be solved without affecting resiliency even if the participants do not know the values of (possibly changing) $n$ and $f$. Specifically, we consider a synchronous system where the participants have unique but not necessarily consecutive identifiers, and give Byzantine agreement algorithms for  reliable broadcast, approximate agreement, rotor-coordinator, early terminating consensus and total ordering in static and dynamic systems, all with the optimal resiliency of $n> 3f$. Moreover, we show that synchrony is necessary as an agreement with probabilistic termination is impossible in a semi-synchronous or asynchronous system if the participants are unaware of $n$ and $f$.
\end{abstract}

%\begin{IEEEkeywords}
%approximate agreement, Byzantine agreement, consensus, reliable broadcast, blockchain, dynamic networks.
%\end{IEEEkeywords}
%\IEEEpeerreviewmaketitle

\section{Introduction}

Many modern networks have to be always available and it may not be possible to
know the size of the network or the number of failures  in advance, since they
may change over time.
Consider, for example, a database cluster that requires frequent node scaling because of changing load, or a wireless sensor network that experiences a changing number of faulty or disconnected nodes over time. Nakamoto's blockchain \cite{bitcoin} is a prominent example where the network is permissionless, i.e., the network is open to \textit{any} number of nodes. So, the number of participants and consequently, the number of failures also change over time. Agreement is a fundamental distributed computing primitive for fault-tolerant networks, however, much of the existing literature assumes that the size \(n\) of the network and/or the upper bound \(f\) on the number of failures  is known to every node \cite{attiya:book, aspnes:notes, charron-bost:approxCon, tseng:conDirectGraphs, mendes:multidimCon}.

In this paper, we consider fault-prone systems where the nodes do not know the number  of nodes \(n\) and the maximum number of Byzantine nodes \(f\) and study fundamental agreement problems for such systems, in particular: 
\begin{itemize}
    \setlength\itemsep{0em}
  \item Reliable broadcast --- ensures that a message is either accepted by every correct node or no correct node \cite{srikanth:reliable};
  \item Rotor-coordinator --- selects \(f + 1\) leaders for the correct nodes;
  \item Consensus --- every correct node has a binary input and the correct nodes output a common binary value that is an input of some correct node \cite{lamport:byz}; 
%  \item Renaming --- the correct nodes choose their new identifiers from a smaller range of identifiers maintaining their uniqueness \cite{attiya:renaming};
    \item Approximate agreement --- each correct node has a real number input and has to output a real number that is \emph{strictly} within the correct input values \cite{approxAgree};
    \item Total ordering --- each correct node maintains a total order on the system events while participants may enter and leave subject to $n>3f$.
\end{itemize}
Since a correct node does not know \(n\) and \(f\) and a Byzantine node may not
announce itself to everyone, there might be more Byzantine nodes in the system
than what a correct node thinks. Thus, it may not be possible to achieve a
resiliency of \(n > 3f\), which can be achieved when the nodes know \(n\) and
\(f\).
When \(f\) is known and the identifiers are consecutive, it is easy to agree on
a set of \(f+1\) nodes, and consequently ensure the presence of a single correct
leader node in the set.
We show, however, that these problems can be solved without affecting
resiliency even when \(n\) and \(f\) are not known.
Specifically, we give algorithms for solving the above problems in synchronous systems with the resiliency of \(n > 3f\), which is optimal for approximate agreement, reliable broadcast and consensus. We also show that the synchronous assumption is necessary as otherwise the problem is impossible and there is non-zero probability of terminating with a disagreement.

An advantage of designing algorithms without the knowledge of \(n\) and \(f\) is
their application to networks where the set of participants change over time. We
illustrate this by extending some of our algorithms to solve Byzantine agreement in dynamic networks. In case of dynamic networks, single shot problems where a node acts on one input and terminates with one output are not very useful. So, we consider an agreement task where nodes are required to totally order the events in a system and design an algorithm for that task.

%We focus on the synchronous model. Many of the most popular protocols in the wild may seem asynchronous at first sight, as they are hiding their synchronous core in fail-over procedures that are fired when a message is not received before a \textit{timeout} happens. Also, we need to distinguish nodes from each other, so we assume authenticated channels. We will explain our model more formally in Section 3.

\section{Related Work}
If the nodes do not known \(n\) and \(f\), then the synchronous assumption is necessary. Otherwise, if the network is asynchronous and the message delays are unbounded, agreement is impossible even with probabilistic termination, as we show later. There is a line of work that deals with this problem using oracles or failure detectors \cite{cavin:unknownParticipants,  greve:knowConnect, alchieri:byzUnknownParticipants}. 
The idea is that a failure detector supplies information about the number of participants. But, these works also assume that every node knows \(f\). In \cite{taheri:unknownDynamicByzCon}, the authors consider an asynchronous dynamic system with a failure detector where \(n\) and \(f\) are unknown, but the failure detector assumed eventually removes the Byzantine nodes. 
%The rotor-coordinator algorithm with unknown $n$ and $f$ has been explored recently as an in-progress work leaving other questions and impossibility results open.

Several ways of improving the robustness of synchronous systems with Byzantine failures have also been explored. For example, Gallet et al.~\cite{gallet:byzAgreeHomo} examine a system that can allocate the same identifier to multiple nodes. In \cite{banu:mobile, garay:byzMobile, bonomi:aaMobile}, the authors examine a synchronous system with mobile Byzantine faults --- those which hop from one node to another across rounds. In \cite{lenzen:ssCount, lenzen:ssClockConsensus}, the authors consider self-stabilizing agreement problems in the presence of Byzantine faults, i.e., the correct nodes have to recover from arbitrary initial state even when the other Byzantine nodes maliciously prevent the correct nodes from recovering. In \cite{garg:weigthed}, the machines are assigned weights and the total weight of the faulty machines is less than a third of the total weight.

The Byzantine agreement problems have a long history since the work by Lamport et al.~\cite{lamport:byz}. They gave a \(f+1\) round algorithm with exponential in \(n\) message complexity for \(n>3f\). They also showed that the resilience of \(n>3f\) is optimal. Berman et al.~\cite{king} later improved the message complexity to polynomial in \(n\), while keeping optimal resilience, and increasing the number of rounds by a small constant. Garay et al.~\cite{garay:optSyncBA} further improved the round complexity to exactly \(f+1\), while retaining optimal resilience and polynomial message complexity. The algorithm by Berman et al.~\cite{king} is well known as the \textit{king} algorithm and is still commonly used \cite{lenzen:ssCount, dolev:synCount, afek:objCon}.
The approximate agreement algorithm was introduced by Dolev et
al.~\cite{approxAgree} and is a useful primitive in designing distributed
algorithms \cite{welch:clockSync, dolve:ssClockSync}.%, vaidya:approxDirected}.
It also requires \(n>3f\) and has optimal resiliency \cite{fischer:easyImposs}. Srikanth et al.~\cite{srikanth:reliable} introduced the reliable broadcast abstraction and its use in dealing with Byzantine failures for \(n>3f\). As they remark, this resiliency is optimal as the reliable broadcast abstraction can be used to solve consensus. The algorithms for approximate agreement, reliable broadcast and consensus in this paper generalize the ones from  Dolev et al.~\cite{approxAgree}, Srikanth et al.~\cite{srikanth:reliable} and Berman et al.~\cite{king} respectively. 

%Okun et al.~\cite{okun:byzRename} introduce the renaming problem with Byzantine failures; the renaming problem itself was introduced by Attiya et al.~\cite{attiya:renaming} for asynchronous systems with crash failures. Okun et al.~show that the problem cannot be solved if \((n + n \bmod 3) \leq 3f\), the number of rounds taken depends only on \(n\) and the Byzantine nodes can forge their identifiers. In the renaming problem, a node is only required to output its new name. It can do more by outputting the new names of all the correct nodes, which we also do in this work, and for this problem it is not hard to show that \(n > 3f\) is the optimal resilience. In \cite{denysyuk:orderRename}, the authors give order-preserving and faster renaming algorithms, namely with \(O(\log n)\) rounds and constant number of rounds, trading it off with a larger range of new identifiers. 

A rotor-coordinator, as also used in~\cite{king}, is an approach to deal with at most \(f\) Byzantine faults by rotating through \(f+1\) coordinator nodes, thus ensuring that one coordinator would be correct. The rotor-coordinator can be easily implemented by rotating through \(f+1\) nodes when \(f\) is known and the identifiers are consecutive. However, it is one of the main bottleneck when \(n\) and \(f\) are unknown and the identifiers are also non-consecutive.

\section{Significance of this Work}
It is not so difficult to observe that if all the correct nodes broadcast in a round, then each correct node \(v\) receives less than \(n_v/3\) messages from the Byzantine nodes --- where \(n_v\) is the number of messages received by the node \(v\) --- irrespective of whether the Byzantine nodes broadcast or not. This observation helps in removing dependency on \(n\) and \(f\) from the classic known algorithms. However, this observation is not sufficient by itself. Many of the classic algorithms run for fixed \(f + 1\) rounds, selecting a different leader in each round. This is a non-trivial problem in our setting, since \(f\) is not a common knowledge and also the identifiers are not consecutive. \sref{alg:rotor} for rotor-coordinator essentially solves this problem. 

The classical models studied in the literature do not allow the Byzantine nodes to lie about the number of participants in the network since it is assumed to be known by every node. Our system model allows the Byzantine nodes to send conflicting information so that the correct nodes never have a consistent information about the number of participants. Therefore, the algorithms designed are robust against a wider range of malicious behavior. This is especially useful for large dynamic systems where it may not be possible to consistently initialize every node with the value of \(n\) and \(f\).

On the other hand, participants are assumed to have access to consistent clocks
after initialization, since the computation is assumed to proceed in rounds. So,
some consistent initialization (synchronization) is still needed. This is
somewhat necessary, since we also show that Byzantine consensus cannot be solved
with probabilistic termination if the system is semi-synchronous or asynchronous and the participants do not know \(n\) and
\(f\). This implies that
it is impossible to build blockchain systems for solving agreement problems in  asynchronous networks when
\(n\) and \(f\)  are not known.
% observation: 
% not suffices
% implies that the byzantine can do more but does not affect the resiliency and the resulting system is more robust 

\section{Model}
The system consists of \(n\) nodes, out of which at most \(f\) are faulty nodes. The faulty nodes can behave in anyway whatsoever, also known as \emph{Byzantine} behavior. We call the non-Byzantine nodes \textit{correct}. The nodes have unique identifiers, which are not necessarily consecutive. Each node knows its identifier only at initialization  apart from a possible input and does not know any global information like \(n\) or \(f\).
The system is \emph{synchronous} and the computation proceeds in \emph{rounds}. In each round, every node receives the messages that were sent to it in the previous round, does some local computations, and then sends again messages to the other nodes to be consumed in the following round. A correct node can broadcast a message to all the nodes or send a message to a specific node that sent a message to the node before. The identifier of a node is included in the message it sends so the receiver of the message can decipher its sender. Thus, a Byzantine node cannot forge its identifier when communicating directly. However, it can help other Byzantine nodes to do so indirectly by claiming to have received messages from other, possible non-existent, nodes. Byzantine nodes can send duplicate messages across rounds but duplicate messages from the same node in a round are simply discarded.

Note that the only way for a correct node to know about the existence of another node is to receive a message from that node. A Byzantine node may get itself known to only a subset of nodes, however, it can behave as if it already knows all the nodes without having received messages from them. In the rest of the paper, we will sometimes refer to the above model as the \emph{id-only} model for brevity. We give the following algorithms  in the \emph{id-only} model for \(n > 3f\):  reliable broadcast in \sref{sec:relBcast}, rotor-coordinator in \sref{sec:rotor}, consensus in \sref{sec:con}, and approximate agreement in \sref{sec:approxAgree}. In \sref{sec:sync}, we show that to solve agreement with probabilistic termination, when \(n\) and \(f\) are unknown, synchronous assumption is necessary. 
In \sref{sec:earlycon}, we give a parallel version of the  consensus algorithm,
where several consensus algorithms can be run in parallel, however, the nodes do
not initally agree on the instances to start.
In \sref{sec:dynamicNetworks}, we build on the parallel consensus to give algorithms for achieving approximate agreement and total ordering of events in a dynamic network.
In \sref{sec:disc}, we discuss the results and some further interesting questions.

\section{Reliable Broadcast}\label{sec:relBcast}

Reliable broadcast \cite{srikanth:reliable} is an abstraction to deal with the messages sent by the Byzantine nodes. The idea is to enforce that a Byzantine node cannot send contradictory information to different nodes. It can still send around false information but the abstraction ensures that the same false information is seen by all the correct nodes. Concretely, let \(s\) be a designated node that may or may not be correct and \((m,s)\) be a message broadcast by \(s\). The message \((m,s)\) is \emph{reliably broadcast} when the following three properties are satisfied. 
\begin{enumerate}
    \item Correctness: If \(s\) is correct, then each correct node accepts \((m,s)\).
    \item Unforgeability: If a correct node accepts a message \((m,s)\) and \(s\) is a correct node, then the message \((m,s)\) was broadcast or sent to all the nodes by the node \(s\).
    %If \(s\) is correct and does not send \((m,s)\), then no correct process accepts \((m,s)\). 
    \item Relay: If a correct node accepts the message \((m,s)\) in a round \(r\), then each correct node accepts the message \((m,s)\) by the round \(r + 1\).
\end{enumerate}

\sref{alg:relBcast} gives an algorithm for a node \(v\) to reliably broadcast a message \((m,s)\) sent by a node \(s\) in the first round. Note that in \sref{ln:nv}, the value \(n_v\) is \emph{not} the number of messages received in the round \(r\) but the number of nodes that sent at least one message to \(v\) until the current round \(r\). Also, the algorithm does not terminate as the idea is to use this mechanism as a subroutine in another algorithm that implements its own termination mechanism, as we will see for consensus, where few additional messages per round are used to detect termination. In the following lemmas, we show that the algorithm satisfies the three properties of the reliable broadcast. We will again assume that \(n > 3f\).

\newcommand{\echo}[1]{\mathit{echo(#1)}}
\begin{algorithm}[!htb] 
\caption{Reliable broadcast algorithm for a node \(v\) to broadcast a message \((m,s)\) sent by a node \(s\) in the first round. Each iteration of the loop is a single round.}\label{alg:relBcast}
\begin{algorithmic}[1]
\Statex
\If{v = s}\Comment{Round 1}
    \State Broadcast \((m,s)\).  \label{ln:init}
\Else
    \State Broadcast \(\mathit{present}\). \label{ln:present}
\EndIf
\If{Received \((m,s)\) from \(s\)}\Comment{Round 2}
\State Broadcast \(\echo{m,s}\). \label{ln:echoInit}
\EndIf
\For{\(r \leftarrow 1\)  to \(\infty\)}\Comment{Rounds 3 to \(\infty\)}
\State  \parbox[t]{\dimexpr\linewidth-\dimexpr\algorithmicindent}{%
  Let \(n_v\) be the number of nodes that sent at least one message to \(v\) until the round \(r\).\strut}\label{ln:nv}
\If{Received at least \(n_v/3\) \(\echo{m,s}\) messages and \\
\hspace{\algorithmicindent}not accepted \((m,s)\) already}
\State Broadcast \(\echo{m,s}\). \label{ln:echoBroad}
\EndIf
\If{Received at least \(2n_v/3\) \(\echo{m,s}\) messages and \\
\hspace{\algorithmicindent}not accepted \((m,s)\) already}
\State Accept \((m,s)\). \label{ln:accept}
\EndIf
\EndFor
\end{algorithmic}
\end{algorithm}

\begin{lem}\label{lem:rb-correct} 
If \(n > 3f\), then \sref{alg:relBcast} satisfies the correctness property of the reliable broadcast.
\end{lem}
\begin{proof}
If the node \(s\) is correct, it sends the message \((m,s)\) to all the nodes during the initial broadcast (\sref{ln:init}).
Every good node receives the message and broadcasts \(\echo{m,s}\) in the next round (\sref{ln:echoInit}).
Let \(g\) be the number of good nodes.
Then, in the third round, every correct node receives \(g\) \(\echo{m,s}\) messages.
Moreover, the value of \(n_v \leq n\) in the third round as \(n\) is the maximum number of nodes that can send a message to \(v\).
As \(n> 3f\), we have \(g > 2f\) or \(3g > 2(f + g) = 2n\). Thus, we have \(g > 2n/3 \geq 2n_v/3\).
Therefore, every correct node accepts the message in the third round (\sref{ln:accept}).
\end{proof}

\begin{lem}\label{lem:rn-g1}
If \(n > 3f\) and a correct node \(v\) receives at least \(n_v/3\) copies of a message \(m\) from distinct nodes in a round \(r\), then at least one of those messages was sent by a correct node.
\end{lem}
\begin{proof}
Let \(f_v''\) be the number of faulty nodes that sent \(m\) to \(v\) in the round \(r\).
Since every correct node transmits a message in the first round (Lines~\ref{ln:init} and \ref{ln:present}), we have \(n_v\geq g\), where \(g\) is the number of good nodes.
So, we can write \(n_v = g + f_v'\), where \(f_v'\) is the number of faulty nodes that sent at least one message to \(v\) until the round \(r\).
Using \(f_v'' \leq f_v'\) and \(n_v = g+ f_v'\), the number of correct nodes \(G\) that sent a message to \(v\) in the round \(r\) are  at least \(n_v/3 - f_v'' \geq (g - 2f_v')/3\). As \(g > 2f\), we have \(G > 2(f - f_v')/3\) or at least one as \(f \geq f_v'\).
So, at least one correct node sent the message \(m\) to \(v\) in the round \(r\).
\end{proof}

\begin{lem}\label{lem:rb-unforge}
If \(n > 3f\), then \sref{alg:relBcast} satisfies the unforgeability property of the reliable broadcast.
\end{lem}
\begin{proof}
We need to show that if a correct node accepts a message \((m,s)\) and \(s\) is a correct node, then the message \((m,s)\) was broadcast by \(s\).
If a message \((m,s)\) was accepted by a correct node \(v\) in a round \(r_a\), then \(v\) received at least \(2n_v/3\) \(\echo{m,s}\) messages in the round \(r_a\).
Thus, the number of correct nodes from which \(v\) received the \(\echo{m,s}\) messages in round \(r_a\) are at least \(2n_v/3 - f_v'' \geq n_v/3 - f_v''\), where \(f_v''\) is the number of messages received by \(v\) from the faulty nodes in the round \(r_a\).
Using \sref{lem:rn-g1}, at least one of the \(\echo{m,s}\) messages received by \(v\) in the round \(r_a\) was sent by a correct node.

Let \(r_f\) be the first round when an \(\echo{m,s}\) message was sent by a correct node \(u\). 
Thus, in the round \(r_f\), the node \(u\) either received at least \(n_{u}/3\) \(\echo{m,s}\) messages or received the message \((m,s)\) from \(s\) (Lines~\ref{ln:echoBroad} or \ref{ln:echoInit}). 
If \(u\) received at least \(n_u/3\) \(\echo{m,s}\) messages, then using \sref{lem:rn-g1}, there is at least one correct node that sent an \(\echo{m,s}\) message in the previous round.
Since \(r_f\) is the first round when a correct node sends an \(\echo{m,s}\) message, the node \(u\) must have received the message \((m,s)\) from \(s\) in the round \(r_f\).
Thus, node \(s\) indeed sent the message \((m,s)\).
As \(s\) is correct, the message \((m,s)\) was broadcast to all the nodes in the first round.
\end{proof}

\begin{lem}\label{lem:rn-g2}
If \(n>3f\) and a correct node \(v\) receives at least \(2n_v/3\) copies of a message \(m\) in a round \(r\), then every correct node \(u\) receives at least \(n_u/3\) copies of \(m\) in the round \(r\).
\end{lem}
\begin{proof}
As \(v\) receives at least \(2n_v/3\) messages, at least \(2n_v/3 - f_v''\) of them were sent by the correct nodes, where \(f_v''\) is the number of messages received by \(v\) from the faulty nodes in the round \(r\).
Let \(f_v'\) be the number of faulty nodes from which \(v\) received at least one message until the round \(r\).
Then, we have \(2n_v/3 - f_v'' = 2(g+f_v')/3 - f_v''\), where \(g\) is the number of good nodes.
As \(f_v'' \leq f_v'\) and \(f_v' \leq f\) by definition, we have  \(2(g+f_v')/3 - f_v'' \geq (2g- f)/3\).

Using \(n > 3f\) or \(g > 2f\), we have \((2g-f)/3 = (g + (g-f))/3 >  (g+ f)/3\). 
Thus, at least \((g+ f)/3\) correct nodes broadcast the  message \(m\) and every correct node receives at least \((g+ f)/3\) copies of \(m\) in the round \(r\).
For a correct node \(u\), we have \((g+ f)/3 \geq (g+f_u)/3 = n_u/3\), where \(f_u\) is the number of faulty nodes from which \(u\) has received at least one message until the round \(r\).
\end{proof}

\begin{lem}\label{lem:rb-relay}
If \(n > 3f\), then \sref{alg:relBcast} satisfies the relay property of the reliable broadcast.
\end{lem}
\begin{proof}
Let \(r\) be the first round in which a correct node \(v\) accepts the message \((m,s)\).
Then, we show that every correct node accepts the message \((m,s)\) by the round \(r + 1\).

As \(v\) accepts the message \((m,s)\) in round \(r\), it received at least \(2n_v/3\) \(\echo{m,s}\) messages.
Using \sref{lem:rn-g2}, each correct node \(u\) receives at least \(n_u/3\) \(\echo{m,s}\) messages in the round \(r\). 
So, every correct node broadcasts \(\echo{m,s}\) message in the round \(r\) (\sref{ln:echoBroad}) and each one of them receives \(g\) \(\echo{m,s}\) messages in the round \(r+1\).
As \(g > 2f\), we have \(3g > 2(f + g) = 2n\). Thus, we have \(g > 2n/3 \geq 2n_u/3\) for every correct node \(u\).
Therefore, every correct node accepts the message \((m,s)\) in the round \(r+1\).
\end{proof}

Using \sref{lem:rb-correct}, \sref{lem:rb-unforge} and \sref{lem:rb-relay}, all the properties of the reliable broadcast are satisfied and we have the following theorem.

\begin{thm}
If \(n > 3f\), then \sref{alg:relBcast} satisfies the properties of the reliable broadcast in the \emph{id-only} model.
\end{thm}

%So, in the reliable a message is accepted again after being broadcast?
% Which property out of relay, unforgeability or correctness is being violated 

\section{Rotor-Coordinator}\label{sec:rotor}
The purpose of Rotor-Coordinator is to have a common coordinator node in each round, where the coordinator node is trusted by everyone in that round. After \(f+1\) different coordinators are selected, everyone is sure that at least one of those \(f+1\) selected coordinators was correct, since there are at most \(f\) faulty nodes. \sref{alg:rotor} gives the algorithm for selecting a set of \(f + 1\) different coordinators, each one in a separate round.

\begin{algorithm}[!htb]
\algnewcommand{\LeftComment}[1]{\Statex\hspace{\algorithmicindent}\(\triangleright\) #1}
\caption{Rotor-Coordinator algorithm for a node \(v\). The sets \(C_v\) and \(S_v\) are used by \(v\) to store process identifiers. The set \(C_v\) is ordered by the process identifiers in increasing order. We use \(|C_v|\) for the size of \(C_v\) and \(C_v[i]\) for its \(i^{\mt{th}}\) member, where \(i \geq 0\). The set \(B_v\) holds messages before they are broadcast by \(v\) at the end of a round. Note that each iteration of the loop is a single round.}
\label{alg:rotor}
\begin{algorithmic}[1]
\State \(C_v \leftarrow \phi\)\Comment{Set of \emph{candidate} coordinators} \label{ln:rc-cdef}
\State \(S_v \leftarrow \phi\)\Comment{Set of \emph{selected} coordinators} \label{ln:rc-sdef}
\State Broadcast \(\mt{init}\).\Comment{Round 1}\label{ln:rc-init}
\State Broadcast \(\mt{echo}(p)\) if received \(\mt{init}\) from \(p\). \Comment{Round 2}
\For{\(r \leftarrow 0 \to \infty\)} \Comment{Rounds 3 up to termination} 
\State \(B_v \leftarrow \phi\)\Comment{\(B_v\) is broadcast at the round's end}
\State \parbox[t]{\dimexpr\linewidth-\dimexpr\algorithmicindent}{%
Let \(n_v\) be the number of nodes that sent at least one message to \(v\) until the round \(r\).\strut}
\If{Received at least \(n_v/3\) \(\mt{echo}(p)\) and \(p \notin C_v\) \\\hspace{\algorithmicindent}}
\State \(B_v \leftarrow B_v \cup \{\mt{echo}(p)\}\) \label{ln:rc-rb1}
\EndIf
\If{Received at least \(2n_v/3\) \(\mt{echo}(p)\) and \(p \notin C_v\) \\\hspace{\algorithmicindent}}
\State \(C_v \leftarrow C_v \cup \{p\}\) \label{ln:rc-rb2}
\EndIf

\State \(p \leftarrow C_v[r \bmod |C_v|]\) \Comment{\(p\) is the next coordinator} \label{ln:rc-nextc}%\Comment{\parbox[t]{0.4\linewidth}{Select the next coordinator as \(p\).\strut}}
 \State \parbox[t]{\dimexpr\linewidth-\dimexpr\algorithmicindent}{%
Let \(p'\) be the coordinator selected in the previous round.\strut
}
\If{Received \(\mt{opinion}(x)\) from \(p'\)}
\State Accept \(x\) as the coordinator's opinion. \label{ln:rc-opnac}
\EndIf
\If{\(p \in S_v\)}
\State \textbf{break} \label{ln:rc-brk}
\EndIf
\State \(S_v \leftarrow  S_v \cup \{p\}\) \label{ln:rc-addtos}
\If{\(v = p\) }\Comment{if \(v\) itself is the coordinator}
\State Let \(o_v\) be \(v\)'s current opinion.
%\Statex\hspace{\algorithmicindent}\hspace{\algorithmicindent}\(\triangleright\) \parbox[t]{\dimexpr\linewidth - 2\dimexpr\algorithmicindent}{Broadcast \(v\)'s opinion at the end of the round.\strut}
%\Statex \Comment{Broadcast \(v\)'s opinion at the end of the round.}
\State \(B_v \leftarrow B_v \cup \{\mt{opinion}(o_v)\}\) \label{ln:rc-opnbr}
\EndIf
\State Broadcast \(B_v\) if its non-empty.
\EndFor
\end{algorithmic}
\end{algorithm}

The idea is that every correct node broadcasts its willingness to become a coordinator initially, when the faulty nodes may or may not participate (\sref{ln:rc-init}). Every correct node \(v\) keeps a set of candidate coordinators \(C_v\), which it updates in a reliable broadcast fashion (Lines \ref{ln:rc-rb1} and \ref{ln:rc-rb2}). In each round, a correct node \(v\) selects the coordinator with the next larger identifier, say \(p\), from the set \(C_v\) and adds it to the set of selected coordinators \(S_v\) (Lines~\ref{ln:rc-nextc} and \ref{ln:rc-addtos}). The node \(v\) accepts the opinion from \(p\) in the next round as the coordinator's opinion (\sref{ln:rc-opnac}) and broadcasts its own opinion as the coordinator's opinion in case \(v\) was selected as the coordinator from the set \(C_v\) (\sref{ln:rc-opnbr}). The node \(v\) terminates when it reselects the same node as the coordinator (\sref{ln:rc-brk}). The hope is that by the time a correct node terminates, it has already witnessed a round in which every correct node accepts the opinion of a common and a correct coordinator. We start by observing that if a correct node adds \(p\) to its set of candidate coordinator \(C_v\), then another correct node \(u\) adds \(p\) to its set \(C_u\) as well by the next round.

\begin{lem}\label{lem:rc-relay}
If a correct node \(v\) adds \(p\) to the set \(C_v\) in a round \(r\), then any correct node \(u\neq v\) adds \(p\) to the set \(C_u\) by the round \(r+1\).
\end{lem}
\begin{proof}
The set \(B_v\) is emptied at the beginning of every round \(r\) and is broadcast at the end of the round \(r\). Thus, the algorithm for adding a process identifier \(p\) to \(C_v\) is same as that of accepting a message \((m,s)\) in \sref{alg:relBcast} if \((m,s) = p\). So, the lemma follows using \sref{lem:rb-relay} for the relay property of the reliable broadcast.
\end{proof}

We call a round a \emph{good round} if the same node \(p\) was selected as a coordinator by every correct node and the node \(p\) is correct. In the following, we show that every correct node witnesses a good round before it terminates, if \(n > 3f\). We will call a round as a \emph{silent} round if the set \(C_v\) remains unchanged for every correct node \(v\), i.e, no correct node executes \sref{ln:rc-rb2} in that round. A \emph{non-silent} round is a round that is not silent. We observe that in a silent round, the value of \(C_v\) is identical for every correct node \(v\). If they were not, then there is a silent round between a correct node \(v\) adding an identifier \(p\) to its set \(C_v\) and another correct node \(u\neq v\) adding \(p\) to its set \(C_u\). This contradicts \(\sref{lem:rc-relay}\). The assumption \(n > 3f\) is used for reliable broadcast and also to ensure a good round. With \(n>4f\), a good round is easily ensured, but \(n>3f\) suffices with careful observation as follows.

\begin{lem}\label{lem:rc-gdrnd}
If \(n > 3f\), then every correct node witnesses at least one good round until it terminates.
\end{lem}
\begin{proof}
Assume for contradiction that a node \(v\) terminates in the round with \(r = r_t\) without witnessing a good round. 
%As a node only terminates when it selects the same node as the coordinator, the node \(v\) accepts a node as the coordinator again without witnessing a good round.
%Consider the first round or the smallest value of \(r\), say \(r_t\), when the node \(v\) selects the same node as the coordinator without witnessing a good round.
Consider a round with \(r = r_c \leq r_t\).
Let \(F_v\subseteq C_v\) and \(G_v \subseteq C_v\), respectively, be the set of faulty node identifiers and the set of good or correct node identifiers in \(C_v\) when the coordinator node is selected in the round \(r_c\) (\sref{ln:rc-nextc}).
Thus, we have \(|C_v| = |F_v| + |G_v|\).

Using \sref{lem:rb-correct}, all the correct node identifiers are added to \(C_v\), even before the first coordinator node is selected.
So, we have \(|G_v| = n - f\) and \(|C_v| = |F_v| + n - f\).
Using \(n > 3f\), we get \(|C_v| > |F_v| + 2f\).
Say that there is \emph{no} correct node \(u\) that added a faulty identifier to its set \(C_u\) in the round with \(r = 0\).
Then, every correct node selects a common coordinator from the set \(G_v\) and \(v\) witnesses a good round before termination, a contradiction.
Thus, there is a correct node \(u\) that adds a faulty identifier to its set \(C_u\) in the round with \(r=0\).
For every non-silent round afterwards, at least one faulty node identifier is added to the set \(C_u\) of some correct node \(u\).
Using \sref{lem:rc-relay}, if a faulty node identifier \(p\) is added to \(C_u\), every correct node \(w \neq u\) adds \(p\) to \(C_w\) by the next round.
Thus, we have \(2f \geq n_\mt{ns}\), where \(n_\mt{ns}\) is the number of non-silent rounds prior to the round \(r_c\) and starting from the round \(r = 0\).
Therefore, we have \(|C_v| > |F_v| + n_\mt{ns}\).

Moreover, until the round \(r_c\), node \(v\) has neither witnessed a good round, nor it has selected the same node again as a coordinator by our assumption.
So, in all the silent rounds prior to the round \(r_c\), a unique faulty node was selected as a coordinator by \(v\).
Therefore, if \(\mt{n_s}\) is the number of silent rounds prior to the round \(r_c\), then \(|F_v| \geq n_s\) since \(v\) selects a node as a coordinator only after adding it to the set \(C_v\).
So, we have \(|C_v| > n_s + n_\mt{ns}\).

Since \(r\) starts from \(0\), we have \(n_s + n_\mt{ns} = r_c\).
So, we have \(|C_v| > r_c\) and \(r_c \bmod |C_v| = r_c\). 
Since the above inequality is true for every round \(r_c \leq r_t\), a node that was already selected as a coordinator, is in the set \(\{C_v[r \bmod |C_v|]: r < r_c\}\).
Therefore, for selecting the same identifier as a coordinator again, it must be that \(r > |C_v| > r_c\), a contradiction.
\end{proof}

\begin{thm}\label{thm:rc}
If \(n > 3f\), then every correct node terminates in \(O(n)\) rounds and there is a round in which every correct node accepts the opinion of a common and a correct coordinator node.
\end{thm}
\begin{proof}
As a node terminates as soon as it selects the same node as a coordinator and there are \(n\) nodes in total, the node terminates in at most \(n\) rounds. Using \sref{lem:rc-gdrnd}, the node also witnesses a good round before termination and accepts the corresponding opinion in the next round (\sref{ln:rc-opnac}).
\end{proof}

\section{Consensus}\label{sec:con}

\newcommand{\mtmx}{\mathit{max}}
\newcommand{\mtmn}{\mathit{min}}
In this section, we give an \(O(f)\) round consensus algorithm in the \emph{id-only} model, where \(f\) is the number of faulty byzantine nodes in the system. \sref{alg:earlyCon} gives an algorithm based on \cite{earlyStopKing}. Every correct node \(v\) has an input \(x_v\), which is a real number. Again, every correct node has to output a common correct value. If the inputs are all same, then the output must be that value. We consider real number inputs here, unlike binary inputs in \sref{sec:con}, since we use it later for ordering events in a system, which can be non-binary.

\begin{algorithm}[!htb]
\caption{An \(O(f)\) round consensus algorithm in the \emph{id-only} model. To initialize the rotor-coordinator in \protect\sref{ln:econInitRC}, run the first two lines of the \protect\sref{alg:rotor}. To initialize \(n_v\) in \protect\sref{ln:initn}, collect the identifiers from which a message has been received, and count them. Later, a node only accepts messages from a node if it counted towards \(n_v\) during the initialization and discards the messages from the other nodes. If a node \(u\) receives a message from another node \(v\) during initialization but not later inside the loop, then \(u\) assumes that \(v\) sent the same message as sent by \(u\) in the previous round. `Next Round' is abbreviated as N.R.}
\label{alg:earlyCon}
\begin{algorithmic}[1]
\Statex
\State Initialize rotor-coordinator. \Comment{Rounds 1 and 2} \label{ln:econInitRC}
\State Initialize \(n_v\). \label{ln:initn}
\While{\(\mt{true}\)}
\State Broadcast \(\mt{input}(x_v)\). \Comment{N.R.}
\If{Received at least \(2n_v/3\) \(\mt{input}(x_v)\)}  \Comment{N.R.}
\State Broadcast \(\mt{prefer}(x_v)\). \label{ln:econPrefer}
\EndIf
\If{Received at least \(n_v /3\) \(\mt{prefer}(x)\)} \Comment{N.R.}
\State \(x_v = x\) \label{ln:econSwitch}
\EndIf
\If{Received at least \(2n_v /3\) \(\mt{prefer}(x)\)}
\State Broadcast \(\mt{strongprefer}(x)\). \label{ln:econStrongprefer}
\EndIf
%\State \parbox[t]{\linewidth - \algorithmicindent - \widthof{ \(\triangleright\) N.R.}}{%
%Execute a round of rotor-coordinator using \(x_v\) as \(v\)'s current opinion. Let \(c\) be the value accepted as the coordinator's opinion.\strut}\Comment{Next Round}\label{ln:econRC}
\State \parbox[t]{\linewidth - \algorithmicindent}{Execute a round of rotor-coordinator using \(x_v\) as \(v\)'s current opinion. Let \(c\) be the value accepted as the coordinator's opinion.\Comment{N.R.}\strut}\label{ln:econRC}
%\State \parbox[t]{200pt}{Run a round rotor-coordinator using \(x_v\) as \(v\)'s opinion. Let \(c\) be the coordinator's opinion.\strut} \label{ln:econRC}
\If{Received less than \(n_v /3\) \(\mt{strongprefer}(x)\)\Comment{N.R.}\\\hspace{\algorithmicindent}}  
\State \(x_v = c\) \label{ln:econSwitchKing}
\EndIf
\If{Received at least \(2n_v/3\) \(\mt{strongprefer}(x)\)}
\State Terminate and output \(x\). \label{ln:econTerminate}
\EndIf
\EndWhile
\end{algorithmic}
\end{algorithm}

In the following, we prove the correctness of \sref{alg:earlyCon}. We refer to an iteration of the loop as a \emph{phase}.

\begin{lem}\label{lem:earlyConValidity}
If \(x_v = x\) for every correct node at the start of the phase, all the nodes terminate with the output \(x\) at the end of the phase.
\end{lem}
\begin{proof}
Every correct node broadcasts \(\mt{input}(x)\) at the start of the phase. So, every correct node \(v\) receives \(g\) \(\mt{input}(x)\) messages. 
As \(n > 3f\), we have \(g > 2f\). 
Thus, we have \(g + 2g > 2(f+ g)\) or \(g > 2n/3 \geq 2 n_v/3\).
So, all the correct nodes broadcast \(\mt{prefer}(x)\) (\sref{ln:econPrefer}). 
Every correct node \(v\) receives \(g \geq 2n_v/3\) \(\mt{prefer}(x)\) messages, keeps their opinion to \(x\) (\sref{ln:econSwitch}), and broadcasts \(\mt{strongprefer}(x)\) (\sref{ln:econStrongprefer}).
Again, each correct node \(v\) receives \(g \geq 2n_v/3\) \(\mt{strongprefer}(x)\) messages and terminates with the output \(x\) (\sref{ln:econTerminate}).
\end{proof}

\begin{lem}\label{lem:quorum}
If a correct node \(u\) receives \(2n_u/3\) copies of a message \(m\) and a correct node \(v\) receives \(2n_v/3\) copies of a message \(m'\) in the same round, then at least one correct node sent both \(m\) and \(m'\) in the previous round.
\end{lem}
\begin{proof}
The number of messages \(G\) sent by the good nodes is at least \(2n_u/3 - f_u + 2n_v/3 - f_v\), where \(f_u\) is the number of \(m\) messages sent to \(u\) by the faulty nodes, and \(f_v\) is the number of \(m'\) messages sent to \(v\) by the faulty nodes.
As \(n_u = g + f_u\) and \(n_v = g + f_v\), we have \(G > 4g/3 - (f_u + f_v)/3\). 
We have \(g > f_u + f_v\) since \(f_u \leq f\), \(f_v \leq f\) and \(g > 2f\). 
Thus, we have \(G > g\) and at least one correct node sent both \(m\) and \(m'\) in the previous round.
\end{proof}

\begin{lem}\label{lem:earlyConTerminate}
If a correct node terminates in a phase, then all other correct nodes have the same opinion at the end of the phase.
\end{lem}
\begin{proof}
Say, a correct node \(v\) terminates with the output \(x\).
Then, it received at least \(2n_v/3\) \(\mt{strongprefer}(x)\) messages.
So, all the correct nodes received at least \(n_v/3\) \(\mt{strongprefer}(x)\) messages using \sref{lem:rn-g2}
and none of them switches to the coordinator's opinion (\sref{ln:econSwitchKing}).
Moreover, at least one correct node \(u\) sent a \(\mt{strongprefer}(x)\) message using \sref{lem:rn-g1}.
The node \(u\) received \(2n_u/3\) \(\mt{prefer}(x)\) messages.
Using \sref{lem:rn-g2}, at least \(n_u/3\) of those messages were sent by the correct nodes and so each node received at least \(n_u/3\) \(\mt{prefer}(x)\) messages.
It is not possible that a correct node also received \(\mt{prefer}(x')\), where \(x \neq x'\).
Indeed, if it was so, then using \sref{lem:rn-g1} there is a correct node \(s\) that sent \(\mt{prefer}(x)\) and a correct node \(t\) that sent \(\mt{prefer}(x')\). 
Thus, node \(s\) received \(2n_s/3\) \(\mt{input}(x)\) messages and node \(t\) received \(2n_t/3\) \(\mt{input}(x')\) messages. 
Using \sref{lem:quorum}, a correct node sent both \(\mt{input}(x)\) and \(\mt{input}(x')\) messages in the same round, a contradiction.
Thus, every correct node changed their opinion to \(x\) (\sref{ln:econSwitch}), which remains unchanged until the end of the phase.
\end{proof}

\begin{lem}\label{lem:earlyConKing}
If the coordinator is correct and none of the correct nodes have terminated, then all the correct nodes have the same opinion by the end of the phase.
\end{lem}
\begin{proof}
Consider the first phase when the coordinator is correct.
Either every correct node \(v\) receives less than \(n_v/3\) \(\mt{strongprefer}(x)\) messages, in which case all the correct nodes have the same opinion by the end of the phase, and we are done. 
Otherwise, there is a correct node \(u\) that received at least \(n_u/3\) \(\mt{strongprefer}(x)\) messages.
Using \sref{lem:rn-g1}, there is at least one correct node \(w\) that sent a \(\mt{strongprefer}(x)\) message.
Thus, node \(w\) received at least \(2n_w/3\) \(\mt{prefer}(x)\) messages.
Using \sref{lem:rn-g2}, every correct node \(v\) received at least \(n_v/3\) \(\mt{prefer}(x)\) messages. 
As before, it is impossible that a correct node \(v\) also receives \(n_v/3\) \(\mt{prefer}(x')\) messages, where \(x \neq x'\).
So, every correct node, including the coordinator, changes its opinion to \(x\) (\sref{ln:econSwitch}).
Since the coordinator is correct, it sends the same opinion \(x\) to all the nodes.
Thus, even if some correct nodes decide to change their opinion to the coordinator's opinion, all the correct nodes still have the same opinion at the end of the phase.
\end{proof}

We can now combine the previous lemmas into the following theorem.

\begin{thm}\label{thm:earlyCon}
\sref{alg:earlyCon} solves consensus in \(O(f)\) rounds in the \emph{id-only} model.
\end{thm}
\begin{proof}
If the correct nodes have the same input \(x\), then they output \(x\) using \sref{lem:earlyConValidity}.
Otherwise, one of the following happens within \(O(f)\) rounds: either a correct node terminates with an output \(x\) or a correct coordinator gets picked.
In either case, the correct nodes have the same opinion at the end of the phase, and terminate with the same output by the next round using \sref{lem:earlyConValidity}.
\end{proof}
\section{Approximate Agreement}\label{sec:approxAgree}
In the approximate agreement problem \cite{approxAgree}, each correct node takes a real number input and outputs a real number. Let \(i_\mtmn\) and \(i_\mtmx\), respectively, be the minimum and the maximum value that is an input of a correct node. Similarly, let \(o_\mtmn\) and \(o_\mtmx\), respectively, be the minimum and maximum value that is output by a correct node. The values output by the correct nodes must satisfy the following conditions.
\begin{enumerate}
    \item The value output by each correct node is within the \emph{input range} \([i_\mtmn, i_\mtmx]\).
    \item The output range  \([o_\mtmn, o_\mtmx]\) is strictly smaller than the input range, i.e., \((o_\mtmx - o_\mtmn) < (i_\mtmx - i_\mtmn)\) if 
    \(i_\mtmx \neq i_\mtmn\).
\end{enumerate}

\begin{algorithm}[!htb]
\caption{Approximate Agreement algorithm for a node $v$. The input value of the node is \(i_v\).}
\label{alg:approxAgree}
\begin{algorithmic}[1]
\Statex
\State Broadcast \(i_v\) to all the nodes (including self).
\State Let \(R_v\) be the set of received values and \(n_v = |R_v|\).
\State Discard \(\lfloor n_v/3\rfloor\) smallest and \(\lfloor n_v/3\rfloor\) largest values from the set \(R_v\) to obtain the set \(S_v\).
\State Output \(o_v=(\min{S_v} + \max{S_v})/2\), where \(\min{S_v}\) and \(\max{S_v}\)  are the minimum and maximum value of the set \(S_v\) respectively.
\end{algorithmic}
\end{algorithm}
\sref{alg:approxAgree} solves the problem. The following lemma shows that the algorithm satisfies the first property of the approximate agreement, i.e., the output range lies within the input range.

\begin{lem}\label{lem:aaWithin}
If \(n > 3f\), then \(o_v \in [i_\mtmn, i_\mtmx]\) for every correct node \(v\).
\end{lem}
\begin{proof}
Let \(g\) be the number of correct nodes.
Then, the node \(v\) receives at least \(g\) values from the correct nodes after the first round.
Let \(f_v\) be the number of values received by \(v\) from the Byzantine nodes.
Therefore, we have \(f_v \leq f\) as \(f\) is the number of faulty nodes and \(v\) receives at most one value from each faulty node in a round.
As \(n = f + g\) and \(f_v \leq f\), we can rewrite \(n > 3f\) as \(g + f > 2f_v + f\). Thus, we have \((g + f_v)/3 > f_v\) or \(\lfloor(g + f_v)/3\rfloor \geq f_v\) as \(f_v\) is an integer. As \(n_v = g + f_v\), we have \(\lfloor n_v/3\rfloor \geq f_v\). 

As there are at most \(f_v\) faulty values in the set \(R_v\) and \(\lfloor n_v/3\rfloor \geq f_v\), the minimum value \(\min{S_v}\) left after discarding \(\lfloor n_v/3\rfloor\) smallest values from \(R_v\) satisfies \(\min{S_v} \geq i_\mtmn\), where \(i_\mtmn\) is the minimum value received from a correct node.
Using a similar argument, the maximum value \(\max{S_v}\) satisfies \(\max{S_v} \leq i_\mtmx\).
Therefore, the output \(o_v\), which is the average of \(\min{S_v}\) and \(\max{S_v}\), satisfies \(o_v \in [i_\mtmn, i_\mtmx]\).  
\end{proof}

\newcommand{\mtmd}{\mathit{med}}
Let \(i_\mtmd\) be the median of the input values at the correct nodes.
In the following lemma, we show that the value \(i_\mtmd\) is never discarded by a correct node while computing the set \(S_v\).

\begin{lem}\label{lem:aaMed}
If \(n > 3f\), then the value \(i_\mtmd \in S_v\) for every correct node \(v\).
\end{lem}
\begin{proof}
Let \(g\) be the number of correct nodes.
Using \(n > 3f\) and \(n = g + f\), we get \(f < g/2\).
Using \(n_v = g + f_v\) and \(f_v \leq f\), we get
\(\lfloor n_v/3 \rfloor  \leq  n_v/3 = (g + f_v)/3 \leq  (g+f)/3\).

As \(f < g/2\), we get \(\lfloor n_v/3\rfloor < g/2\). 
Therefore, even if all the smallest \(\lfloor n_v/3\rfloor\) discarded values are from the good nodes, then also strictly less than half of the smallest good values are discarded to obtain the set \(S_v\). Similarly, strictly less than half of the largest good values are discarded to obtain the set \(S_v\).
Thus, we have \(i_\mtmd \in S_v\).
\end{proof}

Combining the previous two lemmas, we can state the following theorem.

\begin{thm}
If \(n > 3f\), then \sref{alg:approxAgree} achieves approximate agreement in the \emph{id-only} model.
\end{thm}
\begin{proof}
Using \sref{lem:aaWithin}, the output range lies within the input range and the first property of the approximate agreement is satisfied.

Using \sref{lem:aaMed}, we have \(i_\mtmd \in S_v\). Thus, we have \(\max{S_v} \geq i_\mtmd\) and that \(\min{S_v} \leq i_\mtmd\). 
Moreover, using \sref{lem:aaWithin}, we also get that \(\min{S_v} \geq i_\mtmn\) and \(\max{S_v} \leq i_\mtmx\). 
Therefore, we have that the average \(o_v = (\min{S_v} + \max{S_v})/2\) lies within the range \([(i_\mtmn + i_\mtmd)/2, (i_\mtmd + i_\mtmx)/2]\). 
So, the size of output range \((o_\mtmx - o_\mtmn) = (i_\mtmx - i_\mtmn)/2 < (i_\mtmx - i_\mtmn)\) if \(i_\mtmx \neq i_\mtmn\).
\end{proof}

%First, every node \(v\) broadcasts its input \(i_v\) to all the nodes including self. Let \(n_v\) be the number of values that \(v\) receives. Then, the node \(v\) discards \(\lfloor n_v/3 \rfloor\) highest and \(\lfloor n_v/3 \rfloor\) smallest values out of the received ones, and outputs the average of the maximum and the minimum of the ones that are left. Since \(f \leq n_v/3\), one can show that the median correct value is never discarded and the correct output range is at least halved. \sref{app:approxAgree} contains the full analysis.

\section{Synchrony is Necessary}\label{sec:sync}
In our work, we have assumed that the system is synchronous. Intuitively, this is a necessary assumption as a node does not know \(n\) and \(f\) and hence, the number of messages to wait for before deciding. So, it might end up deciding before receiving a message that was delayed for long, as such the decision might be incorrect. The following lemma proves this for consensus.

\begin{lem}
In an asynchronous system where the number of nodes \(n\) and an upper bound \(f\) on the number of failures is not known to the nodes, consensus is impossible, even with probabilistic termination.
\end{lem}
\begin{proof}
Assume a system \(\mathcal{S}\) in which all the nodes are correct.
We partition the set of the nodes into sets \(A\) and \(B\).
A node \(v\) has input \(1\) if \(v \in A\); input \(0\) if \(v \in B\).
The messages between \(A\) and \(B\) are arbitrarily delayed.
To a node \(v \in A\), this is indistinguishable from a system \(\mathcal{A}\) where the nodes in \(B\) are absent, as \(v\) only knows its id initially in both \(\mathcal{S}\) and \(\mathcal{A}\).
Similarly, system  \(\mathcal{S}\) is indistinguishable  to a node \(v \in B\) from a system \(\mathcal{B}\)  where the nodes \(A\) are absent.
The nodes \(A\) decide \(0\) in  the system \(\mathcal{A}\) with a non-zero probability, since they only hear from the nodes with the input \(0\).
Similarly, the nodes \(B\) decide \(1\) in the system \(\mathcal{B}\) with a non-zero probability.
So, the nodes in the system \(\mathcal{S}\) decide on different values with a non-zero probability.
\end{proof}

Similar problems can happen in a \emph{semi-synchronous} system \cite{dwork}, where the message delays have a fixed upper bound $\Delta$, but its value is unknown to the nodes. However, the previous argument does not work since we cannot arbitrarily delay the messages due the existence of the fixed upper bound $\Delta$. Instead, we start with the partitions $\mathcal{A}$ and $\mathcal{B}$ and inductively build an invalidating execution for a union of them.

\begin{lem}
In a semi-synchronous system, where the message delays have a fixed upper bound $\Delta$ and the nodes do not know the value of $\Delta$, $n$ and $f$, consensus is impossible, even with probabilistic termination. 
\end{lem}
\begin{proof}
Consider a system $\mathcal{A}$ where all the nodes have input $1$ and the message delays are at most  $\Delta_a$.
Each node $v \in \mathcal{A}$ decides $1$ with non-zero probability.
Let $E_a$ be such an execution in $\mathcal{A}$ of duration $T_a$.
Similarly, consider another system $\mathcal{B}$ where all the nodes have input $0$ and the message delays are at most $\Delta_b$.
Let $E_b$ be an execution in $\mathcal{B}$ where all the nodes decide $0$ in duration $T_b$.
We consider another system $\mathcal{S}$ consisting of $|\mathcal{A}|+ |\mathcal{B}|$ nodes, and set the maximum message delay $\Delta_s > \max(\Delta_a, T_a, \Delta_b, T_b)$.
We partition the set $\mathcal{S}$ into a set $A$ of $|\mathcal{A}|$ nodes and a set $B$ of $|\mathcal{B}|$ nodes.
The nodes in $A$ have input $1$ where as the nodes in $B$ have input $0$.
We also assume some bijective mapping between the sets $A$ and $\mathcal{A}$ and between the sets $B$ and $\mathcal{B}$.
We use $a'$ denote the counterpart of $a$ in this bijective map.

We construct an execution $E_s$ from $E_a$ and $E_b$ as follows.
If a node $a \in \mathcal{A}$ sends a message to a node $b \in \mathcal{A}$, then $a' \in \mathcal{S}$ sends the same message to $b'$. 
The message sent in $\mathcal{S}$ has the same delay as the message sent in $\mathcal{A}$.
If a node $a \in \mathcal{A}$ \emph{broadcasts} a message to all the nodes $\mathcal{A}$, then $a' \in \mathcal{S}$  broadcasts the same message to all the nodes $\mathcal{S}$.
The delays for the messages broadcast are assigned as follows.
The message delay in $\mathcal{S}$ for the messages broadcast to the nodes $A \subset \mathcal{S}$ is same as the delay of those messages in $\mathcal{A}$.
The message delay in $\mathcal{S}$ for the messages broadcast to the nodes $B \subset \mathcal{S}$ are $\Delta_s$. 
Similarly, we assign message actions and delays to the nodes $B \subset \mathcal{S}$.
Inductively, a node $a \in A \subset \mathcal{S}$ makes the same decisions as a node $a' \in \mathcal{A}$, since both of them do not know the value of $n$ and $f$, and node $a$ makes the (same) decision before it even hears from a node in $B$.
Similarly, a node $b \in B \subset \mathcal{S}$ makes the same decisions as a node $b' \in \mathcal{B}$. 
Therefore, there is an execution $E_s$ in $\mathcal{S}$ so that $a \in \mathcal{S}$ decides $1$ and $b \in \mathcal{S}$ decides $0$, a disagreement.
\end{proof}

The above argument essentially means that an agreement protocol designed to work without the knowledge of $n$ and $f$ (such as the Bitcoin blockchain \cite{bitcoin}), either must assume synchronous execution for guaranteed agreement or sacrifice agreement with some probability. 
%Similar problems can happen in an semi-synchronous system where there exists an up 

\section{Parallel Consensus}\label{sec:earlycon}

\newcommand{\rt}{{:}}
In the consensus problem, each correct node had only one opinion and had to output a single opinion in agreement with other nodes. Later, when a correct node can submit multiple opinions, we need to agree on every opinion submitted by a correct node.
Therefore, we consider the \emph{parallel consensus} problem: Every correct node \(v\) has a set of \(k_v\) input pairs \((\mt{id}_v^i, x_v^i)\) for \(1 \leq i \leq k_v\), where \(x_v^i\) is an opinion and \(\mt{id}_v\) is the identifier of the input pair. Each correct node outputs a set of pairs subject to the following conditions.
\begin{enumerate}
    \item Validity: If \((id, x)\) is an input pair of every correct node and \(x \neq \bot\), then all the correct nodes must output the pair \((id,x)\).
    \item Agreement: If a correct node \(v\) outputs a pair \((id_v, x_v)\), then all other correct nodes must output \((id_v, x_v)\) as well. 
    \item Termination: Every correct node outputs a set of pairs in finite number of rounds.
\end{enumerate}
Note that the rules allow a pair \((id_v, x_v)\) as an input of a correct node \(v\), but not all the correct nodes, and be absent from the output of every correct node. 

 First, we describe the \(\mt{EarlyConsensus(id)}\) algorithm, where every correct node \(v\) has at most one input pair \((\mt{id}, x_v)\), i.e., all nodes may not be aware of the identifier \(\mt{id}\). The pseudocode is given in \sref{alg:earlyConWithId}. To help a node \(v\) distinguish if another node \(u\) is aware of \(\mt{id}\) or has no preference or no strong preference of an opinion, we use \(\mt{id\rt nopreference}\) and \(\mt{id\rt nostrongpreference}\) messages.  

\begin{algorithm}[!htb]
\caption{\(\mt{EarlyConsensus(id)}\) algorithm at node \(v\): The node has at most one input pair \((\mt{id}, x_v)\). The rotor-coordinator and \(n_v\) are initialized as in \protect\sref{alg:earlyCon}. Later, a node only accepts messages from a node if it counted towards \(n_v\) during the initialization and discards the messages from the other nodes. 
The types \(M = \{\mt{id\rt input}, \mt{id\rt prefer}, \mt{id\rt strongprefer}\}\) of received messages are counted as follows. If a message of type \(m\in M\) is received for the \emph{first} time during the second phase, then it is discarded (considered as not received). If a message of type \(m \in M\) is received for the \emph{first} time during the first phase, then the message \(m(\bot)\) is substituted for every node \(u\) that counted towards \(n_v\) during initialization but did not send a type \(m\) message. If a node \(v\) has received a type \(m \in M\) message already during the first phase and a node \(u\) that counted towards \(n_v\) does not send a type \(m' \in M\) message in a subsequent round, then for every such node \(u\),
the node \(v\) substitutes the message of type \(m'\) that it sent most recently.
`Next Round' is abbreviated as N.R.}
\label{alg:earlyConWithId}
\begin{algorithmic}[1]
\Statex
\State Initialize rotor-coordinator. \Comment{Rounds 1 and 2} \label{ln:econIdInitRC}
\State Initialize \(n_v\). \label{ln:econIdinitn}
\While{\(\mt{true}\)}
\If{Input pair \((\mt{id}, x_v)\) present and \(x_v \neq \bot\)}
\State Broadcast \(\mt{id\rt input}(x_v)\). \Comment{N.R.}
\EndIf
\If{Received at least \(2n_v/3\) \(\mt{id\rt input}(x_v)\)}  \Comment{N.R.}\label{ln:econIdInput}
\State Broadcast \(\mt{id\rt prefer}(x_v)\). \label{ln:econIdPrefer}
\Else
\State Broadcast \(\mt{id\rt nopreference}\).  \label{ln:noprefer}
\EndIf
\If{Received at least \(n_v /3\) \(\mt{id\rt prefer}(x)\)} \Comment{N.R.}\label{ln:econIdPrefer1}
\State \(\mt{id}\rt x_v = x\) \label{ln:econIdSwitch}
\EndIf
\If{Received at least \(2n_v /3\) \(\mt{id\rt prefer}(x)\)}\label{ln:econIdPrefer2}
\State Broadcast \(\mt{id\rt strongprefer}(x)\). \label{ln:econIdStrongprefer}
\Else
\State Broadcast \(\mt{id\rt nostrongpreference}\). \label{ln:nostrongprefer}
\EndIf
\State \parbox[t]{\linewidth - \algorithmicindent}{%
Execute a round of rotor-coordinator using \(x_v\) as \(v\)'s current opinion. Let \(c\) be the value accepted as the coordinator's opinion. \Comment{N.R.}\strut}\label{ln:econRC}
%\State \parbox[t]{200pt}{Run a round rotor-coordinator using \(x_v\) as \(v\)'s opinion. Let \(c\) be the coordinator's opinion.\strut} \label{ln:econRC}
\If{Received less than \(n_v /3\) \(\mt{id\rt strongprefer}(x)\)\Comment{N.R.}\\\hspace{\algorithmicindent}} \label{ln:econIdSprefer1}
\State \(\mt{id}\rt x_v = c\) \label{ln:econSwitchKing}
\EndIf
\If{Received at least \(2n_v/3\) \(\mt{id\rt strongprefer}(x)\)}\label{ln:econIdSprefer2}
\State Terminate and output \((\mt{id}, x)\) if \(x \neq \bot\).\label{ln:econIdTerminate}
\EndIf
\EndWhile
\end{algorithmic}
\end{algorithm}

Next, we describe the \(\mt{ParallelConsensus}\) algorithm using the previous one: The node \(v\) starts the \(\mt{EarlyConsensus(id_v)}\) algorithm for every \((\mt{id_v}, x_v)\) pair input at \(v\). If the node \(v\) first hears \(\mt{id\rt input}\), \(\mt{id\rt prefer}\), \(\mt{id\rt strongprefer}\) respectively during the second, third, and fifth round of the first phase and no input pair corresponding to \(\mt{id}\) was present at \(v\), then also the node \(v\) starts the \(\mt{EarlyConsensus(id)}\) algorithm from that round.

\begin{thm}\label{thm:parCon}
The \(\mt{ParallelConsensus}\) algorithm satisfies the parallel consensus properties.
\end{thm}
\begin{proof}
Consider a pair \((\mt{id}, x_v)\) that is input at a correct node \(v\), where \(x_v \neq \bot\).
In the first round of the phase, the node \(v\) broadcasts \(\mt{id\rt input}(x_v)\).
So, every correct node hears an \(\mt{id\rt input}\) message in the second round, and fills the missing opinions from the correct nodes with a \(\mt{id\rt input(\bot)}\).
In the subsequent rounds, if a correct node \(u\) does not receive enough messages to send a \(\mt{id\rt prefer}\) or a \(\mt{id\rt strongprefer}\) message, then it respectively sends a \(\mt{id \rt nopreference}\) and \(\mt{id\rt nostrongpreference}\) message.
So, the node \(v\) does not fill in a message for \(u\).
Therefore, the execution of \(\mt{EarlyConsensus(id_v)}\) is identical to an execution of \sref{alg:earlyCon}, where the input of a correct node \(v\) is \(x_v\) if \((\mt{id_v}, x_v)\) is an actual input and \(\bot\) if such a pair is absent.
Using \sref{thm:earlyCon}, every correct node \(v\) outputs a pair \((\mt{id}, o_v)\) in \(O(f)\) rounds, so that it is in agreement with other correct nodes, and is same as the input \((\mt{id}, x_v)\) if it is present at all the correct nodes.
Discarding the output pairs of the form \((\mt{id}, \bot)\) does not affect the agreement and validity properties required by parallel consensus (\sref{ln:econIdTerminate}).

Now, consider that no correct node has an input pair with the identifier \(\mt{id}\).
If we show that no correct node outputs a pair with the identifier \(\mt{id}\), then we are done.
Let \(r\) be the first round when a correct node \(v\) receives an \(\mt{id}\) message.
If \(r\) is the second phase, or the fourth round (rotor-coordinator) of the first phase, then \(v\) simply discards it. 
Otherwise, the round \(r\) can be the second (\sref{ln:econIdInput}), third (Lines~\ref{ln:econIdPrefer1} and \ref{ln:econIdPrefer2}) or the fifth one (Lines~\ref{ln:econIdSprefer1} and \ref{ln:econIdSprefer2}) of the first phase.
First, consider that \(r\) is the second round of the first phase and a correct node \(v\) first received the \(\mt{id\rt input}\) message during round \(r\).
Since no other correct node \(u\) had an \(\mt{id}\) pair as input, node \(v\) fills a default \(\mt{id\rt input}(\bot)\) for every correct node \(u \neq v\) and decides to broadcast \(\mt{id}\rt prefer(\bot)\).
Similarly, any other correct node \(w \neq v\) that first received the \(\mt{id\rt input}\) message in the round \(r\) broadcasts \(\mt{id}\rt prefer(\bot)\).
In the next round, every correct node receives an \(\mt{id\rt prefer(\bot)}\) message.
If a correct node heard \(\mt{id\rt prefer}\) message for the first time, then it will fill a default \(\mt{id\rt prefer(\bot)}\) for every node \(u\) that did not send a message to it.
If a correct node \(p\) already heard an \(\mt{id}\) message, then we know that it sent \(\mt{id\rt prefer(\bot)}\) in the previous round and will fill the same for missing opinions.
Thus, every correct node \(p\) receives at least \(2n_p/3\) \(\mt{id\rt prefer(\bot)}\) messages, sets \(\mt{id}\rt x_p = \bot\) and broadcasts \(\mt{id\rt strongprefer(\bot)}\).
So, every correct node \(p\) receives at least \(2n_p/3\) \(\mt{id\rt strongprefer(\bot)}\) in the next round, terminates but does not output an \(\mt{id}\) pair since \(\bot\) is the associated opinion.

Now, consider that \(r\) is the third round of the first phase and  a correct node \(v\) first hears an \(\mt{id\rt prefer}\) message in the round \(r\).
The node \(v\) fills a default \(\mt{id\rt prefer(\bot)}\) opinion for every correct node \(u\), sets \(\mt{id\rt x_v = \bot}\) and broadcasts \(\mt{id\rt strongprefer(\bot)}\).
In the next round, every correct node hears an  \(\mt{id\rt strongprefer(\bot)}\) message.
If a correct node \(w\) hears \(\mt{id\rt strongprefer(\bot)}\) for the first time, it fills the missing messages with the default \(\mt{id\rt strongprefer(\bot)}\) message.
If not, the node \(w\) fills the missing opinion with what it sent previously, which is again \(\mt{id\rt strongprefer(\bot)}\).
Thus, every correct node \(w\) receives at least \(2n_w/3\)  \(\mt{id\rt strongprefer(\bot)}\) messages and does not output any \(\mt{id}\) pair.

Lastly, consider that \(r\) is the fifth round of the first phase and a correct node \(v\) first hears an \(\mt{id\rt strongprefer}\) message in the round \(r\). 
No correct node received an \(\mt{id}\) before the round \(r\) by assumption, so no correct node sends an \(\mt{id}\) message before round \(r\).
So, the node \(v\) fills the default \(\mt{id\rt strongprefer(\bot)}\) message for every correct node \(u\).
Consequently, the node \(v\) receives \(2n_v /3 \) \(\mt{id\rt strongprefer(\bot)}\) messages and does not output an \(\mt{id}\) pair.
\end{proof}

\section{Application to Dynamic Networks}\label{sec:dynamicNetworks}
In this section, we see how the protocols that we developed can be applied to networks, where the participants enter or leave the system, subject to the constraint that \(n > 3f\). 
First, we look into the approximate agreement problem. We use \sref{alg:approxAgree} in the dynamic setting as well. 
It is easy to observe that the Lemmas~\ref{lem:aaWithin} and \ref{lem:aaMed} apply even if the participants enter and leave the system in every round subject to \(n > 3f\). So, the range of correct values still gets halved in every round, with respect to the previous round. However, new nodes entering the system might also increase the range of values at the correct nodes. So, whether the range decreases or increases over time depends on the actual inputs of nodes entering or leaving the system.

Next, we consider the problem of total ordering of events in a dynamic system. 
We can run the parallel consensus algorithm in every round to agree on the events occurred during that round.
We just need to make sure that the set of identifiers used for every parallel consensus instance remains consistent.
To do that, we have to specify some more details about the model. 
The adversary can decide the number of nodes that can join the network \emph{before} every round starts, subject to the constraint that \(n > 3f\) remains true when the round starts.
Once a node joins the network, it can broadcast to all the nodes that have joined but not left already.
A node leaves the network by announcing so to all the participants.
A correct node decides itself when to leave.
The adversary decides when a faulty node leaves the network.
\sref{alg:orderDynamic} lists the pseudocode.

% so, what is the dynamic system model.
% Initially, if two nodes join at the same time, can they hear each other?
% Well, lets assume that they cannot,
% Then, if they are part of 
\begin{algorithm}[!htb]
\caption{Algorithm at a node \(v\) to order events in a dynamic network. Initially, round \(r\) is initialized to \(0\) and \(S = \{v\}\). Since there could be multiple parallel consensus instances running at the same time, we identify them by the round in which they start by appending the round number to the messages. 
Also, running a parallel consensus instance with respect to \(S\) means recording the value of \(S\) at the start of the instance, and only accepting the messages from the node identifiers in \(S\), discarding the rest.}
\label{alg:orderDynamic}
\begin{algorithmic}[1]
\Statex
\If{\(v\) wants to participate}
\State Broadcast \(\mt{present}\). \Comment{Next Round}
\State \parbox[t]{\linewidth - \algorithmicindent}{
Let \(A_v\) be the multiset of \((\mt{ack}, t)\) messages received by \(v\) in the next round, where \(t \geq 0\). \strut}
\State \parbox[t]{\linewidth - \algorithmicindent}{ 
Initialize \(r = r_0 + 1\), where \((\mt{ack}, r_0)\) is the majority in \(A_v\).\strut}
\State \parbox[t]{\linewidth - \algorithmicindent}{ Initialize \(S\) to the identifiers which sent a message in \(A_v\). \strut}
\EndIf
\While{\(\mt{true}\)}
\State \(r \leftarrow r + 1\)
\State \(I_v^r \leftarrow \{\}\)
\If{Received \(\mt{present}\) from \(u\)}
\State \(S \leftarrow S \cup \{u\}\)
\State Send \((\mt{ack}, r)\) to \(u\). \Comment{Next Round}
\EndIf
\If{\(v\) wants to stop participating}
\State Broadcast \(\mt{absent}\). \Comment{Next Round}
\State 
\parbox[t]{\linewidth - \algorithmicindent}{ 
Wait and participate in the outstanding parallel \\consensus instances until termination. 
\strut}
\EndIf
\If{Received \(\mt{absent}\) from \(u\)}
\State \(S \leftarrow S \backslash \{u\}\)
\EndIf
\If{\(v\) witnesses an event \(m \neq \bot\)}
\State Broadcast \((m, r)\). \Comment{Next Round}
\EndIf
\If{Received \((m,r-1)\) from \(u\)}
\State \(I_v^r \leftarrow  I_v^r \cup \{(u,m)\}\)
\EndIf
\State
\parbox[t]{\linewidth - \algorithmicindent}{
Start a parallel consensus instance \(r\) with the input pairs \(I_v^r\) with respect to the set \(S\).\Comment{Next Round}\strut}\label{ln:toConsensus}
\State A round \(r' < r\) is final if \(r - r' > 5|S_v^{r'}|/2 + 2\).
\State 
\parbox[t]{\linewidth - \algorithmicindent}{
Let \(R\) be the largest round such that all the rounds at most \(R\) are final. \label{ln:toR}
\strut}
\State 
\parbox[t]{\linewidth - \algorithmicindent}{
Order the outputs of the consensus instances with identifiers at most \(R\) in the order of increasing identifiers, breaking ties arbitrarily. \label{ln:toOutput}
\strut}
\EndWhile
\end{algorithmic}
\end{algorithm}

In the following, we show that the nodes agree on the sequences that they output in \sref{ln:toOutput}. Let \(T_v^r\) be the sequence output by a correct node \(v\) at the end of round \(r\) (\sref{ln:toOutput}). Our goal is that \(T_v^r\) satisfies the following two agreement properties.
\begin{enumerate}
    \item Chain-prefix: For any pair of correct nodes \(u,v\), either \(T_u^r\) is a prefix of \(T_v^r\) or \(T_v^r\) is a prefix of \(T_u^r\).
    \item Chain-growth: For every correct node \(v\), events are appended to \(T_v^r\) over time, if a correct node submits an event in every round. 
\end{enumerate}

\begin{thm}
\sref{alg:orderDynamic} outputs a chain of events that satisfy the chain-prefix and chain-growth properties.  
\end{thm}
\begin{proof}
Initially, the node \(v\) stores the correct round number \(0\). 
By assumption, we have \(n > 3f\) in every round.
Then, by induction on rounds, selecting the round number based on the majority of received \(\mt{ack}\) messages always returns the correct round number for every correct node.
Therefore, every correct node that starts a parallel consensus instance in a round \(r\), tags it with the same identifier \(r\).
Each of these instances are then correct using \sref{thm:parCon}.

Consider a round \(r'\) that is final with respect to \(v\).
Since each phase of \sref{alg:earlyConWithId} is five rounds and the initialization is two rounds, the parallel consensus instance \(r'\) terminates by \(r' + 5 f_r' + 2 \) rounds using \sref{thm:parCon}, where \(f_r'\) is the number of faulty nodes in the round \(r'\).
Let \(g_r'\) be the number of good nodes in the round \(r'\) and \(n_r'\) be the total number of nodes in the round \(r'\).
Since we have \(n_r' > 3f_r'\) by assumption, we have \(|S_v^{r'}| \geq g_r' > 2f_r'\). 
Since \(r'\) is final, the current round \(r > r' + 5|S_v^{r'}|/2 + 2 > r' +5f_r'+2\).
So, the parallel consensus instance \(r'\) has terminated by the previous round and no further output from consensus instance \(r'\) is produced.
Moreover, using \sref{thm:parCon}, any other correct node \(u\neq v\) has also accepted the same output pairs corresponding to the consensus instance \(\mt{id}\). 
Also, node \(u\) has not accepted any other output pairs corresponding to the consensus instance \(\mt{id}\), which would contradict the agreement property of parallel consensus.
Let \(R_u\) and \(R_v\) respectively be the value of \(R\) computed in \sref{ln:toR} by the nodes \(u\) and \(v\).
Then, rounds up to \(R_\mt{min} = \min\{R_u, R_v\}\) are final for both the nodes \(u\) and \(v\).
Thus, the outputs of the consensus instances up to \(R_\mt{min}\) is the common prefix of \(T_u^r\) and \(T_v^r\), which is the common-prefix property. 

Since the parallel consensus instance \(r'\) terminates in \(O(f_r')\) rounds, the earliest non-final round eventually becomes final and the chain-growth property is satisfied as well.
\end{proof}

\iffalse
Building upon the previous algorithms, we can easily devise agreement algorithms for the dynamic networks where the participants enter and leave the system under the influence of an adversary. In the Appendix, we give the algorithms for approximate agreement and total ordering in dynamic networks. The dynamic version of approximate agreement is essentially the same as the static one as the required properties continue to hold in the dynamic case as well. 

The total ordering algorithm makes the correct nodes output a sequence or chain of values such that they satisfy the following properties: (1) If $a$ and $b$ are the current chains of any two correct nodes, then either $a$ is a prefix of $b$ or $b$ is a prefix of $a$. (2) The chain length at every correct node increases over time provided that the a correct node submits a value in every round.  The idea is to use an instance of parallel consensus in every round by tagging these messages with the starting round number of the instance. The full algorithm and the proof is in the Appendix. 
\fi

\section{Discussion}\label{sec:disc}
% forging identifiers might be possible in reliable broadcast as the nodes need to identify the nodes across different rounds.

% why ability to broadcast to all the nodes? Otherwise, partition in the network and we cannot do much.

% joining the system, broadcast abstraction is one way, other way: talk to another node, he gives you the list of identifiers to broadcast

%recap results
In this paper, we investigated distributed systems where the participants are
neither aware of the size $n$ nor the safe estimate $f$ of Byzantine
failures. We examined fundamental distributed computing problems such as,
approximate agreement, reliable broadcast, rotor-coordinator and consensus; concluding that all of them can be solved with the optimal resiliency of \(n > 3f\). 
Each of these algorithms illustrated a different method of computing. 
It is interesting to note that ``replacing'' \(f\) by \(n_v/3\) works in these algorithms although \(n_v/3\) is an incorrect upper bound on the number of failures. 
An algorithm using a combination of some of the discussed primitives could be ``compiled'' to work without the knowledge of \(n\) and \(f\) keeping resiliency unaffected. 
We evaluated resiliency in this work but other metrics such as message complexity, round complexity, etc. do not change much either. For example, the message complexity of reliable broadcast is unaffected compared to the original algorithm, the convergence rate of the approximate agreement algorithm remains unchanged 
and the \(O(f)\) round complexity of consensus algorithm is optimal \cite{fischer:roundLbConsensus}.

%The rotor-coordinator has no previous counterparts, like what would be the ideal comparison, say against the version where n and f are known, and one has to select a set of f + 1 nodes, what would be the round or message complexity? 
Removing knowledge of \(n\) and \(f\) from the participants has other benefits too.
For example, we show in \sref{sec:dynamicNetworks} that the design of agreement algorithms for dynamic networks becomes much easier and the nodes do not need to agree on the number of participants in the network.
It also opens up ways to achieve agreement in networks without using information from every node. For example, consider a set of nodes that are in approximate agreement with each other already and a new node joins. Then, the new node can execute \sref{alg:approxAgree} only with a subset of nodes to get closer to the value of most of the nodes.
Self-stabilizing algorithms may not need to restore the value of \(n\) and \(f\). 

It is unclear if the resiliency of rotor-coordinator is optimal, a question left for further work. Also, one could look if these techniques could benefit semi-synchronous or asynchronous dynamic systems where the rate of change of \(n\) is controlled, since without having any knowledge about $n$ or $f$ guaranteed agreement is impossible in such systems.

%\vspace{0.17cm}
%\noindent
\section{Acknowledgments} 
We would like to thank Christoph Lenzen for the discussions, reading the draft and suggesting improvements.
\bibliography{byz}

\end{document}